\documentclass[11pt]{amsart}

\usepackage{graphicx}
\usepackage{amscd}
\usepackage{amsmath}
\usepackage{amsfonts}
\usepackage{amssymb}
\usepackage{a4wide}
\usepackage{setspace}
\usepackage{enumerate}         
\usepackage{color}
\usepackage{url}
\usepackage{amsthm}
\usepackage{hyperref}
\usepackage{bm}
\usepackage{xy}
\usepackage{stmaryrd}

\theoremstyle{plain}
\newtheorem{theorem}{Theorem}[section]

\newtheorem{corollary}[theorem]{Corollary}

\newtheorem{lemma}[theorem]{Lemma}

\newtheorem{proposition}[theorem]{Proposition}

\newtheorem{definition}[theorem]{Definition}

\newtheorem{assumption}[theorem]{Assumption}
\theoremstyle{remark}
\newtheorem{remark}[theorem]{Remark}

\numberwithin{equation}{section}

\newcommand{\ind}{1\!\kern-1pt \mathrm{I}}
\newcommand{\rsto}{]\!\kern-1.8pt ]}
\newcommand{\lsto}{[\!\kern-1.7pt [}

\numberwithin{equation}{section}

\newcommand{\RR}{\mathbb{R}}

\newcommand{\NN}{\mathbb{N}}

\newcommand{\cA}{\mathcal{A}}

\newcommand{\cC}{\mathcal{C}}
\newcommand{\cD}{\mathcal{D}}

\newcommand{\cF}{\mathcal{F}}
\newcommand{\cM}{\mathcal{M}}

\newcommand{\cZ}{\mathcal{Z}}

\newcommand{\massP}{\mathbf{P}}
\newcommand{\massQ}{\mathbf{Q}}
\newcommand{\massE}{\mathbf{E}}

\newcommand{\lmd}{\lambda}
\newcommand{\conv}{\textnormal{conv}}

\newcommand{\Var}{\operatorname{Var}}

\newcommand{\sint}{\stackrel{\mbox{\tiny$\bullet$}}{}}

\begin{document}
\title[Utility maximization, random endowment, transaction costs]{Utility maximization problem with random endowment and transaction costs: when wealth may become negative}
 \author{Yiqing Lin}
\address{Yiqing Lin\newline\indent Centre de math\'ematiques appliqu\'ees,\newline\indent \'Ecole Polytechnique, Route de Saclay, F-91128 Palaiseau Cedex - France}
  \email{yiqing.lin@polytechnique.edu}
   \author{Junjian Yang}
 \address{Junjian Yang\newline\indent Fakult\"at f\"ur Mathematik, Universit\"at Wien\newline \indent Oskar-Morgenstern Platz 1, A-1090 Wien - Austria}
  \email{junjian.yang@univie.ac.at}

\begin{abstract}
  In this paper we study the problem of maximizing expected utility from the terminal wealth with proportional transaction costs and random endowment.
  In the context of the existence of consistent price systems, we consider the duality between the primal utility maximization problem and the dual one, 
    which is set up on the domain of finitely additive measures. In particular, we prove duality results for utility functions supporting possibly negative values. 
  Moreover, we construct the shadow market by the dual optimal process and consider the utility based pricing for random endowment. 
\end{abstract}

\date{\today}
\thanks{The authors gratefully acknowledge financial support from the Austrian Science Fund (FWF) under grant P25815 and from the European Research Council (ERC) under grant 321111. This work is partially done during the visit of J. Yang at CMAP, \'Ecole Polytechnique, which is very much appreciated.}
\maketitle 


\section{Introduction}
 A classical problem in mathematical finance is how an economic agent maximizes her expected utility from terminal wealth by trading in a financial market.  
 In particular, to consider this problem with general market models beyond Markovian asset prices, a modern approach called ``duality method'' or ``martingale method'' has been developed since 1990s. 
 This approach is based on duality characterizations of portfolios provided by the set of ``martingale measures''. 
 The main idea is to solve a dual variational problem and then to find the solution of the original utility maximization problem by convex duality. 
 In \cite{KS99}, Kramkov and Schachermayer discussed this problem in general semimartingale framework by establishing an abstract duality theory for the primal and dual problems on bipolar subsets of $L^0$. 
 In particular, the authors considered an agent, endowed with deterministic initial wealth, whose preference is modeled by a utility function supporting only positive wealth. 
 The case that utility functions supporting possibly negative wealth was considered by Schachermayer in \cite{Sch01} when the stock price process is a locally bounded semimartingale. 
 Afterwards, Biagini and Frittelli \cite{BF05} generalized the result of \cite{Sch01} by removing the local boundedness assumption. \\
 
 
 The model discussed in \cite{KS99} was subsequently improved in Cvitani\'c et al.~\cite{CSW01} by allowing for bounded random endowment rather than the deterministic one. 
 In that paper, the utility function is identical with the one in \cite{KS99}. However, the dual problem is defined on the enlarged domain of finitely additive measures. 
 Then, the utility maximization problem in \cite{Sch01} with the new setting for the economic agent was treated by Owen \cite{Owe02} when the market is driven by locally bounded semimartingales. 
 More recently, the results of \cite{CSW01} and \cite{Owe02} have been extended in several directions, e.g., relaxing the boundedness assumption on the random endowment or on the stock price process, see \cite{HK04, OZ09, BFG11}.\\
 

 In this paper, we shall consider the utility maximization problem with transaction costs, which is essentially as old as its frictionless counterpart. 
 Cvitani\'c and Karatzas are the first who brought the duality method to this domain. 
 In \cite{CK96}, the problem is set up with utility functions on the positive half line and the market is modeled with a bond and a stock driven by the It\^o process. 
 Moreover, the agent faces constant proportional transaction costs when trading and she is required to liquidate her portfolio to the bond at the end. 
 The authors of \cite{CK96} proved the existence of solutions to the problem of utility maximization by a priori assuming the existence of dual optimizer in a proper domain, 
   whereas a complete result without such assumption was eventually given by Cvitani\'c and Wang \cite{CW01}. 
 Afterwards, Bouchard extended this result by considering utility functions defined on the whole real line and by adding bounded random endowment in \cite{Bou02}. 
 In the frictionless case, we usually assume that there is a single consumption asset, which is used as a num\'eraire. Mathematically whether the agent liquidates her stock holding makes no difference in the frictionless market but does matter in the problem with transaction costs. 
 Therefore, it is quite natural to allow the agent to have access to several consumption assets and this induces a new formulation of the financial model. 
 This new model describes a multi-currency market and was first introduced by Kabanov \cite{Kab99} based on the concept of solvency cone. 
 Based on this model, utility maximization problem with a multivariate utility function was later studied by Deelstra et al.~\cite{DPT01}, Campi and Owen \cite{CO11} without random endowment,
   and by Benedetti and Campi \cite{BC12} with bounded random endowment. 
 They provided static duality results in different ways. \\
 
 
 In contrast with the multi-currency model, we work in the present paper with a continuous-time num\'eraire-based market model driven by a general stock price process. 
 Notice that this process is not necessarily a semimartingale, which is required in the frictionless case for ensuring the absence of arbitrage. 
 In this framework, utility maximization problems have been examined by Czichowsky and Schachermayer in \cite{CS15, CS16b} under the existence of consistent price systems. 
 Precisely speaking, the consistent price system consists of two processes: a fictitious stock prices process, which lives between the bid and ask ones; a (local) martingale density of this fictitious price process. 
 This creative concept introduced by e.g., \cite{JK95, Sch04, GRS08} serves as the equivalent martingale measure in the frictionless case. 
 Under such assumption, Guasoni et al.~\cite{GRS10} and Schachermayer \cite{Sch14b}  proved a superreplication theorem similar to the classical one, 
   so that the primal and dual domains considered in \cite{CS15, CS16b} have a similar functional structure as the ones in \cite{KS99, Sch01}. 
 Therefore, the abstract theorem derived from \cite{KS99, Sch01} could be employed in this new context to solve the utility maximization of the above mentioned both kinds of utility functions. 
 One aim of the present paper is to generalize the result of \cite{CS16b} by bringing in bounded random endowment. 
 To achieve this, we first have an intermediate duality result for the problem on the positive half line, which could be proved by proceeding the argument in \cite{CSW01}. 
 This intermediate duality result is similar to \cite{BC12}, however it is more straightforward and adapted to the num\'eraire-based setting, which is necessary for the subsequent approximation. 
 For the problem on the whole real line, we first construct auxiliary  primal and dual functions by proper truncation in order to come back to the case for the intermediate result. 
 Then, we exhibit similar procedures as in \cite{Owe02} to complete the proof by approximating both optimizers and expected value functions. \\
 

 As introduced above, each consistent price system models a fictitious market. 
 An important question in the theory of portfolio optimization with proportional transaction costs is whether or not there exists a so-called shadow price process, i.e., a least favorable frictionless market extension, 
   that leads to the same optimal strategy and utility. 
 If the answer is affirmative, the behavior of a given economic agent can be explained by passing to a suitable frictionless shadow market. 
 Starting from \cite{KM10}, shadow prices have proved to be useful for solving concrete utility maximization problems, see e.g., \cite{GGMS14, GMS13, HP15}.  
 Moreover, Kallsen and Muhle-Karbe \cite{KM11} have pointed out that shadow prices always exist for utility maximization problems in finite probability spaces. For general utility functions on the positive half line, 
   Cvitani\'c and Karatzas \cite{CK96} worked within the It\^o framework and construct a shadow market whenever the dual solution is a martingale measure. 
 In the light of \cite{CK96}, Czichowsky et al.~\cite{CSY15} showed that, 
   if the stock price process is continuous and satisfies the condition $(NUPBR)$ of ``No Unbounded Profit with Bounded Risk'', the optimizer of the corresponding dual problem is always a local martingale, 
   hence the shadow price exists. 
 In addition, the affirmative result is obtained once certain constraint is confined to the agent. 
 For example, Loewenstein \cite{Loe00} confirmed that shadow prices exist for continuous bid-ask price processes whenever short positions are ruled out. 
 This result is recently generalized by Benedetti et al.~\cite{BCKMK13} with Kabanov's general multi-currency market models and by Gu et al.~\cite{GLY16b} with positive random endowment. 
 However, several counterexamples have been found to show that shadow prices in the classical sense may fail to exist without further assumptions, see \cite{Rok13, CMKS14, BCKMK13, CS15, CSY15}.  
 In the general c\`adl\`ag framework, Czichowsky and Schachermayer introduced a new notion in \cite{CS15} which makes it possible to interpret the dual optimizer as a shadow price but in a generalized sense, regardless whether or not the dual optimizer is a local martingale.
 This generalized shadow price process is defined via a ``sandwiched'' couple consisting of a predictable and an optional strong supermartingale, and pertains to all strategies which remain solvent under transaction costs. 
 This notion is afterward extended by Bayraktar and Yu to the case with random endowment in \cite{BY15}. Inspirited by Hugonnier and Kramkov \cite{HK04}, the authors of \cite{BY15} considered a utility maximization problem with transaction costs and additionally with unbounded random endowment confined by maximal trading strategies and a uniform integrability condition. They constructed generalized shadow prices, whenever the duality result holds and a sufficient condition on the dual optimizer is assumed. \\
 
   
 In contrast to \cite{CS15, BY15}, if we consider the utility maximization problem with utility functions defined on the whole real line, the picture changes. 
 In \cite{CS16b}, Czichowsky and Schachermayer confirmed that the existence of strictly consistent price systems (see \cite{Sch04,CS06} for a definition) with ``finite entropy'' guarantees the existence of classical shadow prices. 
 We shall find in the present paper that the presence of bounded random endowment will not alter this result, which is based on the fact that the dual optimizer is associated with a strictly positive martingale density. This property does not necessarily hold true for the case on positive half line in \cite{CS15, BY15}. However, the case with unbounded random endowment similar to \cite{BY15} is left for future research. 
 Moreover, we provide a more generalized definition of shadow price, i.e., 
   we only require that the shadow price market yields the same optimal utility, then such kind of shadow price could be always constructed from the dual optimizer. \\
   
 %
%
 
 The remainder of the article is organized as follows. In the next section we introduce basic settings and formulate the problem. 
 In Section 3 and Section 4 we successively study the utility maximization problems on the positive real line and on the whole real line and provide respectively the duality results.  
 We discuss in Section 5 the existence of shadow prices and the utility based pricing with exponential utility functions as an application.

\section{Financial market model}
In this section, we briefly review preliminaries of a general continuous-time num\'eraire-based market model with proportional transaction costs. 
Readers interested in more details about this framework can consult \cite{Sch14a, Sch14b, CS15, CS16b, CSY15}.\\

We consider a financial market consisting of one riskless and one risky assets, where the price of the riskless asset is constant and normalized to $1$.
Moreover, trade is permissible over a finite time interval $[0, T]$. Denote by $S=(S_t)_{0\leq t\leq T}$ the stock price process, which is based on a filtered probability space $(\Omega, \cF, (\cF_t)_{0\leq t\leq T}, \massP)$ 
 satisfying the usual hypotheses of right continuity and saturatedness, where $\cF_0$ is assumed to be trivial.  
 
 \begin{assumption} \label{Sassumption}
  The process $S=(S)_{0\leq t\leq T}$ is adapted to $(\cF_t)_{0\leq t\leq T}$, with strictly positive and c\`adl\`ag paths. 
 \end{assumption}
 
 We introduce proportional transaction costs $\lambda> 0$ for trading stock.
 The couple of processes $((1-\lambda)S_t,S_t)_{0\leq t\leq T}$ model the bid and ask prices of stock share,  
 which means that the agent has to pay a higher ask price $S_t$ when buying but only receives a lower bid price $(1-\lambda)S_t$ when selling them. 
 For obvious economic reasons, we assume $\lambda<1$.\\
 
 Consider an agent in this market, endowed with initial wealth $x\in\RR$, who receives moreover an exogenous endowment $e_T$ at time $T$, 
  which is $\cF_T$-measurable and satisfies $\rho:=\|e_T\|_{\infty}<\infty$. \\

 The agent's \textbf{trading strategies} are modeled by $\RR^2$-valued, predictable processes $\varphi=(\varphi^0_t,\varphi^1_t)_{0\leq t\leq T}$ of finite variation, 
 where $\varphi^0_t$ and $\varphi^1_t$ denote the holdings in units of the riskless and the risky asset, respectively, after rebalancing the portfolio at time $t$.
 We note that each process $\varphi$ of finite variation is l\`adl\`ag and can be decomposed into two nondecreasing processes $\varphi^\uparrow$ and $\varphi^\downarrow$ both null at zero, i.e., 
  $\varphi_t=\varphi_0+\varphi^\uparrow_t-\varphi^\downarrow_t$. 
 The total variation $\Var_t(\varphi)$ of $\varphi$ is then given by $\Var_t(\varphi)=\varphi^\uparrow_t+\varphi^\downarrow_t$.
 We denote by $\varphi^c$ its continuous part 
   $$ \varphi^c_t:=\varphi_t-\sum_{s<t}\Delta_+\varphi_s -\sum_{s\leq t}\Delta\varphi_s, $$
   where $\Delta_+\varphi_s:=\varphi_{s+}-\varphi_s$ and $\Delta\varphi_s:=\varphi_s-\varphi_{s-}$ are its right and left jumps, respectively. 

\begin{definition} [self-financing]
  A trading strategy $\varphi = (\varphi_t^0,\varphi_t^1)_{0\leq t\leq T}$ is called \textbf{self-financing under transaction costs $\lambda$}, if 
   \begin{equation} \label{SF}
     \int_s^td\varphi^0_u \leq -\int_s^tS_ud\varphi_u^{1,\uparrow} + \int_s^t(1-\lambda)S_ud\varphi_u^{1,\downarrow},\quad 0\leq s\leq t\leq T,
   \end{equation}
  where the integrals are defined via
   \begin{align*}
     \int_s^tS_ud\varphi_u^{1,\uparrow} &:= \int_s^tS_ud\varphi_u^{1,\uparrow,c} + \sum_{s<u\leq t}S_{u-}\Delta\varphi_u^{1,\uparrow} + \sum_{s\leq u<t}S_u\Delta_+\varphi_u^{1,\uparrow}, \\
     \int_s^tS_ud\varphi_u^{1,\downarrow} &:= \int_s^tS_ud\varphi_u^{1,\downarrow,c} + \sum_{s<u\leq t}S_{u-}\Delta\varphi_u^{1,\downarrow} + \sum_{s\leq u<t}S_u\Delta_+\varphi_u^{1,\downarrow}.
   \end{align*}  
\end{definition}

 The self-financing condition \eqref{SF} states that purchases and sales of the risky asset are accounted for in the riskless position: for $0\leq s< t\leq T$, 
 \begin{align*} 
   \int_s^td\varphi^{0,c}_u & \leq -\int_s^tS_ud\varphi^{1,\uparrow,c}_u+\int_s^t(1-\lambda)S_ud\varphi^{1,\downarrow,c}_u, \notag \\
   \Delta\varphi^0_t&\leq-S_{t-}\Delta\varphi^{1,\uparrow}_t +(1-\lambda)S_{t-}\Delta\varphi^{1,\downarrow}_t,\notag   \\
   \Delta_+\varphi^0_t&\leq-S_t\Delta_+\varphi^{1,\uparrow}_t +(1-\lambda)S_t\Delta_+\varphi^{1,\downarrow}_t.
 \end{align*}

 \begin{definition}[liquidation value]
   We define the \textbf{liquidation value} at time t by 
  \begin{align*}
    V_t^{\mathrm{liq}}(\varphi) := \varphi^0_t + (\varphi^1_t)^+(1-\lambda) S_t - (\varphi^1_t)^-S_t. 
  \end{align*}
 \end{definition}
 
 \begin{remark}
   The following formula can be deduced by integration by parts:
    $$ V_t^{\mathrm{liq}}(\varphi)= \varphi^0_0 + \varphi^1_0S_0 + \int_0^t\varphi^1_udS_u - \lambda\int_0^tS_{u}d\varphi^{1,\downarrow}_u -\lambda S_t(\varphi^1_t)^+. $$
 \end{remark}
 
 \begin{definition}[admissibility]\label{admi}
   A self-financing trading strategy $\varphi$ is called \textbf{admissible}, if there exists $M>0$ such that  for every $[0,T]$-valued stopping time $\tau$,
    $$ V_{\tau}^{\mathrm{liq}}(\varphi) \geq -M, \quad a.s. $$ 
 \end{definition}
 
 For $x\in\RR$, we denote by $\cA^{\lambda}_{adm}(x)$ the set of all self-financing and admissible trading strategies under transaction costs $\lambda$ 
   starting from  $(\varphi^0_0,\varphi^1_0)=(x,0)$ and 
   $$ \cC^{\lambda}(x) := \left\{V^{\mathrm{liq}}_T(\varphi) \,\Big|\, \varphi = (\varphi^0,\varphi^1)\in\cA^{\lambda}_{adm}(x) \right\}. $$
 \begin{definition}[$\lambda$-consistent price system]
     Fix $0<\lambda <1$ and the stock price process $S=(S_t)_{0\le t\le T}$.
     A $\lambda$-\textbf{consistent price system} is a two-dimensional strictly positive process $Z=(Z^0_t,Z^1_t)_{0\le t\le T}$ 
      with $Z^0_0=1$, that consists of a martingale $Z^0$ and a (local) martingale $Z^1$ under $\massP$ such that there exists an adapted process process $\widetilde{S}:=(\widetilde{S})_{t\in [0, T]}$ satisfying for $0\leq t\leq T$,
      \begin{equation}\label{J10}
       \widetilde{S}_t \in [(1-\lambda)S_t, S_t]\quad \mbox{and}\quad  Z^1_t=Z^0_t\widetilde{S},\quad a.s.
      \end{equation}
     The collection of all $\lambda$-consistent price systems is denoted by $\mathcal{Z}^{\lambda}_e(S)$.
     In addition, we denote by $\mathcal{Z}^{\lambda}_a(S)$ the set of all nonnegative processes $Z$ satisfying all conditions above except for the strict positivity.\\
     
     For $0<\lambda<1$, we say that a price process $S=(S_t)_{0\leq t\leq T}$ satisfies $(CPS^\lambda)$, if there exists a $\lambda$-consistent price system. 
      \end{definition}
   
 \begin{remark}
   The presence of transaction costs enables us to consider optimization problems with models beyond semimartingales in an arbitrage-free way. 
   In this new context $\lambda$-consistent price systems play the role of equivalent martingale measures in the frictionless case. 
   For each $Z\in \mathcal{Z}^\lambda_e(S)$, the couple models a fictitious market with stock price process $\widetilde{S}$ that satisfies no free lunch with vanishing risk $(NFLVR)$.
 \end{remark}
   
 According to Schachermayer \cite{Sch14b}, the important superreplication theorem can be established under the following assumption.
 \begin{assumption}  \label{CPSassumption}
   $S$ satisfies $(CPS^{\mu})$ for all $\mu\in(0,1)$. 
 \end{assumption}
 
 \begin{theorem}[superreplication theorem, {\cite[Theorem 1.4.]{Sch14b}}] \label{superreplication}
  Let $S$ satisfy Assumption \ref{Sassumption} and Assumption \ref{CPSassumption}. Fix $0<\lambda<1$. 
  Let $ g\in L^0(\Omega,\cF,\massP)$ be a random variable bounded from below, i.e., $g\geq -M$, a.s., for some $M>0$. 
  Then, $g\in\cC^{\lambda}(x)$, if and only if $\massE[Z^0_Tg] \leq x$, for each $Z\in \mathcal{Z}^\lambda_e(S)$.
\end{theorem}

\section{Utility maximization problem on the positive real line}

 In this section, we suppose the agent's preferences over terminal wealth are modeled by a utility function $U:(0,\infty)\to\RR$, 
   which is strictly increasing, strictly concave, continuously differentiable and satisfies the Inada condition:
     $$ U'(0) := \lim_{x\to 0}U'(x) = \infty \quad\textnormal{ and } \quad U'(\infty) := \lim_{x\to \infty}U'(x) = 0. $$
 Without loss of generality, we may assume $U(\infty)>0$ to simplify the analysis. 
 Define also $U(x)=-\infty$ whenever $x\leq 0$.

 \begin{assumption} \label{U(x)assumption}
  The utility function $U$ satisfies the reasonable asymptotic elasticity, i.e.,
   $$ AE(U):= \limsup_{x\to\infty}\frac{xU'(x)}{{U}(x)}< 1. $$
 \end{assumption}

 We refer the reader to \cite{KS99} for financial interpretation and more results about the previous assumption.\\

 For the utility maximization problem, we restrict our attention to the terminal liquidation wealth, for $x>0$, the primal problem is
   \begin{align} \label{PP}
     \massE[U(x+V^{\mathrm{liq}}_T(\varphi)+ e_T)] \to \max !, \qquad \varphi:= (\varphi^0,\varphi^1) \in \cA^{\lambda}_{adm}(0).
   \end{align}

 We denote $\cC^{\lambda}:=\cC^{\lambda}(0)$. Without loss of generality, we can rewrite $\cC^\lambda$ by
  \begin{equation} \label{equacl} 
    \cC^{\lambda} = \left\{\varphi^0_T \,\Big|\, \varphi = (\varphi^0,\varphi^1)\in\cA^{\lambda}_{adm}(0),\ \varphi^1_T=0 \right\}. 
  \end{equation}                       
 Then the value function of the problem \eqref{PP} is defined as follows
   \begin{equation} \label{PP1}
     u(x):= \sup_{g\in\widetilde{\cC}^{\lambda}}\massE[U(x+g+e_T)],
   \end{equation}
   where the set $\widetilde{\cC}^{\lambda}$ consists of those elements of $\cC^{\lambda}$ for which the above expectation is well defined.  
 Finally, in order to exclude trivial case, we have the following assumption, which implies that $u(x)<\infty$ due to the concavity of $u$.
  \begin{assumption} \label{u(x)assumption}
    The value function $u(x)$ is finitely valued for some $x>\rho$. 
  \end{assumption}

 Let us denote $V: \RR_+\to\RR$ the convex conjugate function of $U(x)$ defined by 
  $$ V(y):=\sup_{x>0}\{U(x)-xy\}, \quad y>0. $$
 It is obvious that $V(y)$ is strictly decreasing, strictly convex and continuously differentiable and satisfies $$ V(0)=U(\infty), \quad V(\infty)=U(0). $$
 We also define $I:(0,\infty)\to (0,\infty)$ the inverse function of $U'$ on $(0,\infty)$, which is strictly decreasing, and satisfies $I(0)=\infty$, $I(\infty)=0$ and $I = -V'$.\\
 
 To consider the dual problem to \eqref{PP1}, the usual dual space is
  $$ \cM^{\lambda}_a:= \left\{Z_T^0\in L^1 \,\Big| \, (Z^0,Z^1)\in\cZ^{\lambda}_a(S) \right\}, $$
   which is a subset of $L^1$. 
 According to \cite{CSW01}, this subset is relatively small to hold the dual optimizer of the problem subsequently defined. 
 Thus, we extend it by completion on the enlarge space $ba=(L^{\infty})^{*}$, the dual space of $L^{\infty}$, and define the following subset of $ba$, which is equipped with the weak-star topology $\sigma(ba,L^{\infty})$,
    $$ \cD^{\lambda}:=\big\{Q\in ba_+\, \big| \, \|Q\|=1 \,\textnormal{ and }\, \langle Q, g \rangle\leq 0,\ \textnormal{ for all }\, g\in\cC^{\lambda}\cap L^{\infty} \big\},  $$
   and $\cD^{\lambda,r}:=\cD^{\lambda}\cap L^1$, where $r$ stands for regular. 
 We note that $\cD^{\lambda}$ is clearly convex and also $\sigma(ba,L^{\infty})$-compact by Alaoglu's theorem. 
 Since $-L^\infty_+\subseteq \cC^\lambda$, we see that $\cD^\lambda\subseteq ba_+$.\\
 
 For each $Q\in ba_+$, it admits a unique Yoshida-Hewitt decomposition in the form $Q = Q^r+Q^s$, where the regular part $Q^r$ is the countably additive part and  
    $Q^s$ is the purely finitely additive part. 
 Moreover, we define for $X\in L^0$, bounded from below, and $Q\in ba_+$, 
    $$ \langle Q,X\rangle:=\lim_{n\to\infty}\langle Q,X\wedge n\rangle\in [0,\infty]. $$
 Then, we observe that for each $ g\in\cC^{\lambda}$, $\langle Q,g\rangle\leq 0$, for all $g\in\cC^{\lambda}$ and $Q\in\cD^{\lambda}$.\\
  
  
 Now we define the dual optimization problem by 
  \begin{equation} \label{defvy}
    v(y):= \inf_{Q\in\cD^{\lambda}}\left\{\massE\left[V\left(y\frac{dQ^r}{d\massP}\right)\right]+y\langle Q,e_T\rangle \right\}.
  \end{equation}


The following theorem is the counterpart of \cite[Theorem 3.1]{CSW01} in the frictionless case.
As \cite{CS15}, the presence of transaction costs does not alter the functional structure of the primal and dual domains.

\begin{theorem}   \label{mainthm1}
  Under Assumptions \ref{Sassumption}, \ref{CPSassumption}, \ref{U(x)assumption}, \ref{u(x)assumption},
  we have
  \begin{enumerate}
   \item $u(x)<\infty$ for all $x\in\RR$ and $v(y)<\infty$ for all $y>0$.
   \item The primal value function $u$ is continuously differentiable on $(x_0,\infty)$ and $u(x)=-\infty$ for all $x<x_0$, 
         {where $x_0:=-v'(\infty)=\sup_{Q\in\cD^{\lambda}}\langle Q,-e_T\rangle$}.
         The dual value function $v$ is continuously differentiable on $(0,\infty)$. 
   \item The functions $u$ and $v$ are conjugate in the sense that 
        \begin{equation*}   \label{v=u-xy}
            v(y) = \sup_{x>x_0}\{u(x)-xy\}, \quad y>0,
        \end{equation*}
        \begin{equation*}   \label{u=v+xy}
            u(x) = \inf_{y>0}\{v(y)+xy\}, \quad x>x_0.
        \end{equation*}
   \item For all $y>0$, there exists a solution $\widehat{Q}_y\in\cD^{\lambda}$ to the dual problem, which is unique up to the singular part. 
         For all $x>x_0$, $\widehat{g}:= I\Big(\widehat{y}\frac{d\widehat{Q}^r_{\hat{y}}}{d\massP}\Big)-x-e_T$ is the solution to the primal problem, 
         where $\widehat{y}=u'(x)$, which attains the infimum of $\{v(y)+xy\}$. 
  \end{enumerate}
\end{theorem}

\begin{proof}
 The proof of this theorem shall be developed in the same way as \cite{CSW01}, thus we omit the details but emphasize where the new results for transaction costs apply. The readers who are interested in the detailed discussion are referred to \cite{GLY16c}.\\


 First of all, we observed that $v$ is finitely value by recalling the property of the dual problem in \cite{CS15} and the boundedness of the random endowment.
 Then, pick up a minimizing sequence $(Q_n)_{n\in \mathbb{N}}$ for the problem \eqref{defvy}. 
 Since $\mathcal D^\lambda$ is convex and weak star compact, we can apply \cite[Lemma A.1]{CSW01} to construct $\widehat{Q}$ by a cluster point of $(\widetilde{Q}_n)_{n\in \mathbb{N}}$, 
   a subsequence of convex combinations of $(Q_n)_{n\in \mathbb{N}}$, such that
   $$ \frac{d\widehat{Q}^r_y}{d\massP} = f = \lim_{n\to\infty}\frac{d\widetilde{Q}^r_n}{d\massP}\quad \mbox{and} \quad \langle \widetilde{Q}_n, e_T\rangle \longrightarrow \langle \widehat Q, e_T\rangle.$$ 
 Thanks to \cite[Lemma 3.2]{KS99}, the negative part of the sequence $\left\{V\left(y\frac{d\widetilde{Q}_n^r}{d\massP}\right)\right\}_{n\in\NN}$ is uniform integrable and thus Fatou's lemma applies for the proof of the optimality of $\widehat{Q}$. 
 It is obvious that the value function is convex due to the convexity of $V$ and moreover, the dual optimizer $\widehat{Q}$ is unique up to the singular part. \\
 

 The differentiability of the value function $v$ and its quantitative properties can be deduced as \cite[Lemma 4.2, Lemma 4.3]{CSW01}. 
 In particular, we have for $y>0$,
   \begin{equation} 
     \label{vu} v'(y)= -\left\langle \widehat{Q}_y^r, I\left(y\frac{d\widehat{Q}_y^r}{d\massP}\right)\right\rangle + \langle \widehat{Q}_y, e_T\rangle. 
   \end{equation}
 Then, for each $x>x_0:=-v'(\infty)$, there exists a unique $\widehat{y}>0$, such that $v'(\widehat{y})+x=0$, and $\widehat{y}$ attains the infimum of $\{v(y)+xy\}$.
 For simplicity, denote ${\widehat Q}:={\widehat Q}_{\widehat y}$.
 Let us consider $$\widehat{g}:=I\left(\widehat{y}\frac{d\widehat{Q}^r}{d\massP}\right)-x-e_T. $$
 It follows from \eqref{vu} that 
   \begin{equation}\label{-x=-Qr,x+X+QseT}
     \begin{aligned} 
       -x = v'({\widehat y}) = -\left\langle \widehat{Q}^r, I\left(\widehat{y}\frac{d\widehat{Q}^r}{d\massP}\right)\right\rangle + \left\langle \widehat{Q}, e_T  \right\rangle 
                             = -\left\langle \widehat{Q}^r, x+ \widehat{g}\right\rangle + \left\langle \widehat{Q}^s, e_T  \right\rangle. 
     \end{aligned}  
   \end{equation} 
  Similar as \cite[Lemma 4.4]{CSW01}, we can prove
    $$ \sup_{Q\in\cD^{\lambda}}\left\{\langle Q^r,x+\widehat{g}\rangle - \langle Q^s,e_T\rangle \right\}= \langle \widehat{Q}^r,x+\widehat{g}\rangle - \langle\widehat{Q}^s,e_T\rangle = x, $$
    which implies that $\langle Q, x+\widehat{g}\rangle\leq x$, hence $\langle Q, \widehat{g}\rangle\leq 0$, for each $Q\in \mathcal{D}^r$. 
  In particular, we obtain that $\massE[Z^0_Tg]\leq 0$, as $\{Z^0_T| Z\in\mathcal Z^\lambda_e(S)\}\subseteq \mathcal{D}^r$. 
  Obviously, $\widehat{g}$ is bounded from below, then by Theorem \ref{superreplication} we obtain that $\widehat{g}\in \mathcal{C}^\lambda$.\\
 %
  
  Due to the strict positivity of $I(\cdot)$, we know $x+\widehat g+e_T>0$, thus
    \begin{equation*}
     \begin{aligned}
          \langle \widehat{Q}, e_T\rangle + x = \langle \widehat{Q}^r, x+\widehat{g}+e_T\rangle \leq \langle \widehat{Q}, x+\widehat{g}+e_T\rangle  
                                              \leq  \langle \widehat{Q}, e_T\rangle + \langle \widehat{Q}, x\rangle \leq \langle \widehat{Q}, e_T\rangle + x.
     \end{aligned}
    \end{equation*}
  Then, from the conjugate property between $U$ and $V$ and the definition of $\widehat{g}$, we have 
 \begin{align*}
  u(x)\geq \massE[U(x+\widehat{g}+e_T)] &=\massE\left[V\left(\widehat{y}\frac{d\widehat{Q}^r}{d\massP}\right)+\widehat{y}\frac{d\widehat{Q}^r}{d\massP}(x+\widehat{g}+e_T)\right] \\
                                        &= \massE\left[V\left(\widehat{y}\frac{d\widehat{Q}^r}{d\massP}\right)\right] + \widehat{y}\langle \widehat{Q},e_T\rangle + x\widehat{y} \\
                                        &=v(\widehat{y})+x\widehat{y}\geq u(x),
 \end{align*}
  where the last inequality could be proved by the superreplication theorem under transaction costs (Theorem \ref{superreplication}). 
  The proof is complete. 
\end{proof}

   

\section{Optimal investment when wealth may become negative}
 In this section we consider the problem with a utility function $U:\RR\to\RR$, which is defined and finitely valued everywhere on the real line. 
 We take the usual assumptions that $U$ is continuously differentiable, strictly increasing, strictly concave and satisfies Inada conditions: 
  \begin{align*}
    U'(-\infty):=\lim_{x\to -\infty}U'(x) = \infty \quad \mbox{ and } \quad U'(\infty):=\lim_{x\to\infty}U'(x)= 0. 
  \end{align*}
  
 We also assume that the function $U$ has reasonable asymptotic elasticity as defined in \cite{Sch01}. 
 
 \begin{assumption}  \label{RAE+-}
  The function $U:\RR\to\RR$ satisfies the reasonable asymptotic elasticity, i.e., 
      \begin{align}\label{RAER}
        AE_{-\infty}(U):=\liminf_{x\to -\infty}\frac{xU'(x)}{U(x)}>1 \quad \mbox{ and } \quad AE_{+\infty}(U):=\limsup_{x\to\infty}\frac{xU'(x)}{U(x)}<1.
      \end{align}
 \end{assumption}

 Our aim of this section is to study the optimization problem  
  \begin{equation}  \label{PPP}
    \massE[U(x+g+e_T)] \to\max!,  \qquad g\in\cC^{\lambda},
  \end{equation}
  where $\cC^\lambda$ is defined in the previous section as the collection of all admissible liquidation values at $T$ (see Definition \ref{admi} and (\ref{equacl})).
  Then,  the corresponding value function $g$ is given by 
   $$ u(x):= \sup_{g\in\cC^{\lambda}}\massE[U(x+g+e_T)]. $$
 As pointed out in \cite{CS16b}, once the utility function supports negative wealth, the optimum of (\ref{PPP}) 
 may not be attained in $\cC^\lmd$ even without random endowment (compare with \cite{Sch01, Owe02} in the frictionless case).
 Hence, similarly to \cite{CS16b, Bou02}, we consider the optimization problem (\ref{PPP}) over an enlarged set $\cC^{\lambda}_U$ defined as below:
  \begin{eqnarray*}
    \cC^{\lambda}_{U}:= \left\{g\in L^0(\massP; \RR\cup\{\infty\})\, \Bigg| \,  \begin{array}{r}
                                                                                 \exists \{g_n\}_{n\in \mathbb{N}}\subseteq \cC^{\lambda}\,\mbox{ s.t. }\, U(x+g_n+e_T)\in L^1(\massP) \mbox{ and } \\
                                                                                 U(x+g_n+e_T) \xrightarrow{L^1(\massP)} U(x+g+e_T)
                                                                              \end{array}
                        \right\}. 
  \end{eqnarray*}
It is obvious that the enlargement of primal domain will not alter the optimal value, that is, 
   \begin{equation}  \label{PPP_Cu}
     \massE[U(x+g+e_T)] \to\max!,  \qquad g\in\cC_U^{\lambda},
   \end{equation} 
yields $$ u(x) = \sup_{g\in\cC^{\lambda}}\massE[U(x+g+e_T)] = \sup_{g\in\cC^{\lambda}_U}\massE[U(x+g+e_T)]. $$
In particular, the fact that $U(x+g_n+e_T)\xrightarrow{L^1(\massP)}U(x+g+e_T)$ implies  $g_n\to g$ in $\massP$, 
since $U$ is strictly increasing. \\

 
In order to rule out trivial cases, we shall make the following assumption, which ensures that $u(x)$ is finitely valued.  
  \begin{assumption}  \label{u(x)U(infty)}
    The value function satisfies $u(x)<U(\infty)$, for some $x\in\RR$. 
  \end{assumption}
 To formulate the dual problem of \eqref{PPP}, we introduce the conjugate function of $U(x)$: 
  $$ V(y):=\sup_{x\in\RR}\big(U(x)-xy\big), \qquad y>0, $$
  which is a continuously differentiable, strictly convex function satisfying 
   $$ V(0)=U(\infty), \, V(\infty)=\infty, \, V'(0)=-\infty, \, V'(\infty)=\infty. $$
 We also have the formula $$ V(y)= U\big(I(y)\big)-yI(y), $$
  where $I$ is the inverse function $(U')^{-1}$, which equals to $-V'$. \\
%
  
 Without loss of generality we assume that $U(0)>0$ after possibly adding a constant to $U$. 
 This implies the strict positivity of $V(y)$ which ensures the results \cite[Corollary 4.2]{Sch01}.

 Now, we are in a position to define the dual problem.
  \begin{equation} \label{DDP}
    v(y):= \inf_{(Z^0,Z^1)\in\cZ^{\lambda}_a(S)}\massE\left[V(yZ^0_T)+yZ^0_Te_T\right] = \inf_{Z^0_T\in\cM^{\lambda}_a}\massE\left[V(yZ^0_T)+yZ^0_Te_T\right].
  \end{equation}

 \begin{remark}
   For all $g\in\cC^{\lambda}$, $y>0$ and $(Z^0,Z^1)\in\cZ^{\lambda}_a(S)$, by the superreplication theorem under transaction costs (Theorem \ref{superreplication}), we have 
    $$ \massE[U(x+g+e_T)] \leq \massE\left[V(yZ_T^0)+ yZ_T^0(x+g+e_T)\right],  $$
    and therefore 
    $$ u(x) \leq \inf_{y>0}\{v(y)+xy\}. $$
 \end{remark}

Here below is the main result of this paper: 

 \begin{theorem}  \label{maintheoremR}
   Under Assumptions \ref{Sassumption}, \ref{CPSassumption}, \ref{RAE+-} and \ref{u(x)U(infty)} and moreover that $S$ is locally bounded, we have 
   \begin{enumerate}
    \item The value functions $u$ (respectively, $v$) is finitely valued, strictly concave (respectively, convex), continuously differentiable function defined on $\RR$ (respectively, $\RR_+$).
          The functions $u$ and $v$ are conjugate and satisfy
           $$ u'(-\infty)=\infty, \quad u'(\infty)=0, \quad v'(0)=-\infty, \quad v'(\infty)=\infty. $$
          Moreover, the function $u$ has reasonable asymptotic elasticity. 
    \item For $y>0$, the optimal solution $\widehat{Z}^0_T(y)\in\cM^{\lambda}_a$ to the dual problem \eqref{DDP} exists and is unique. 
          The map $y\mapsto \widehat{Z}_T^0(y)$ is continuous in the variation norm. 
    \item For $x\in\RR$, the optimal solution $\widehat{g}(x)$ to the primal problem \eqref{PPP} exists in $C^\lambda_U$, which is unique and satisfies 
            $$ x+\widehat{g}(x)+e_T = I\left(\widehat{y}\widehat{Z}^0_T(\widehat{y})\right),  $$
            where $\widehat{y}=u'(x)$. 
    \item We have the formulae for marginal utility:
            \begin{align*}
              v'(y) &= \massE\left[\widehat{Z}^0_T(\widehat{y})\left(V'\big(\widehat{y}\widehat{Z}^0_T(\widehat{y})\big)+ e_T\right)\right];  \\
              u'(x) &= \massE\left[U'\left(x+\widehat{g}(x)+e_T\right) \right];  \\
             xu'(x) &= \massE\left[\big(x+\widehat{g}(x)\big)U'\left(x+\widehat{g}(x)+e_T\right) \right]. 
            \end{align*}
   \end{enumerate}
 \end{theorem}
 
In the light of \cite{Sch01, Owe02}, the proof of this theorem consists of successive  approximations. 
We first construct an increasing sequence $(U_n)_{n\in\NN}$, such that for each $n\in\NN$,
  \begin{itemize}
    \item $U_n=U$ on $[-n, \infty)$;
    \item $-\infty< U_n\leq U$ on $(-(n+1), -n)$;
    \item $U_n=-\infty$ on $(-\infty, -(n+1)]$;
    \item $U_n$ is increasing, strictly concave, continuously differentiable on $(-(n+1), \infty)$, and satisfies
           $$ \lim_{x\to -(n+1)}U_n(x)=-\infty, \qquad \lim_{x\to -(n+1)}U_n'(x)=\infty.$$
  \end{itemize}
  
 Define for $y\geq 0$,
    $$ V_n(y):=\sup_{x\in\RR}\left\{U_n(x)-xy\right\} = U_n\left(I_n(y)\right)-yI_n(y),$$ 
   where $ I_n:=\left(U_n'\right)^{-1}= -V_n'$. \\

Without loss of generality, we could choose the sequence  $(U_n)_{n\in\NN}$, such that there exists a constant $C$ independent of $n$ simultaneously valid for the estimates of $V$ and all $V_n$ in \cite[Corollary 4.2]{Sch01}.\\
%

Define $\widetilde{U}_{n}(\widetilde{x}):= U_{n}\big(\widetilde{x}-(n+1)\big)$, which is a finitely valued for $\widetilde{x}>0$ and satisfies Inada condition at $0$ and $+\infty$, 
  and the reasonable asymptotic elasticity condition at $+\infty$. 
On the one hand, we consider the following utility maximization problem, for $\widetilde{x}>\widetilde{x}_0$,
 \begin{equation}\label{primalmod}
   \widetilde{u}_{n}(\widetilde{x}):=\sup_{g\in\cC^{\lambda}}\massE\left[\widetilde{U}_{n}(\widetilde{x}+g+e_T)\right], 
 \end{equation}
 which has a unique optimal solution $\widetilde{g}_n(\widetilde{x})\in\cC^{\lambda}$. On the other hand, the dual problem of (\ref{primalmod}) is formulated by
\begin{equation}  \label{tildev_n}
   \widetilde{v}_n(y):=\inf_{Q\in\cD^{\lambda}}\left\{\massE\left[\widetilde{V}_n\left(y\frac{dQ^r}{d\massP}\right)\right]+y\langle Q,e_T\rangle\right\},
 \end{equation}
 where $\widetilde{V}_n$ is the conjugate of $\widetilde{U}_n$. Applying Theorem \ref{mainthm1} in the previous section, we know that 
 for ${y}>0$, there exists a unique solution $\widehat{Q}_n(y)\in\cD^{\lambda}$ to the problem \eqref{tildev_n}, and moreover, for $\widehat{y}=\widetilde{u}_n'(\widetilde{x})$,
     \begin{equation*}      \widetilde{g}_n(\widetilde{x})= -\widetilde{V}_n'\left(\widehat{y}\frac{d\widehat{Q}^r_n(\widehat{y})}{d\massP}\right)- \widetilde{x}-e_T.
    \end{equation*}
 Denote $\widetilde{x}:=x+n+1$. We shift back the utility maximization problem by defining
  \begin{equation}  \label{u_nPP}
     u_n(x) :=  \sup_{g\in\cC^{\lambda}}\left[U_n(x+g+e_T)\right]= \widetilde{u}_{n}(\widetilde{x}). 
 \end{equation} 
 It is obvious that the unique solution to the above problem $\widehat{g}_n(x):=\widetilde{g}_n(\widetilde{x})$ and moreover $u'_n(x)=\widetilde u'_n(\widetilde x)$.  Then, the conjugate of $u$ is given by 
  \begin{equation}
    \label{vmod} v_n(y) = \inf_{Q\in\cD^{\lambda}}\left\{\massE\left[\widetilde{V}_n\left(y\frac{dQ^r}{d\massP}\right)\right]+y\langle Q,e_T\rangle\right\} + (n+1)y. 
  \end{equation}
 Taking into account that $V_n(y)=\widetilde{V}_n(y) + (n+1)y$, we have 
  \begin{equation}\label{vder}
    \begin{aligned}
      v_n(y) &= \inf_{Q\in\cD^{\lambda}}\left\{ \massE\left[V_n\left(y\frac{d{Q}_n^r(y)}{d\massP}\right)\right]+y\left\langle {Q}_n(y),e_T\right\rangle 
                     + (n+1)y\left(1-\massE\left[\frac{d{Q}_n^r(y)}{d\massP}\right]\right)\right\}.
    \end{aligned}
  \end{equation}
 Summing up, for $\widehat{y}=u'_n(x)=\widetilde{u}_n'(\widetilde{x})$, the unique solution $\widehat{Q}_n(\widehat{y})$ solving \eqref{tildev_n} is also a solution to \eqref{vmod} and \eqref{vder}. 
 Besides, the solution to \eqref{u_nPP} admits the following representation
  \begin{equation}  \label{x+g_n+e_T=-Vn'(yQnr)}
      x+\widehat{g}_n(x)+e_T = -V_n'\left(y\frac{d\widehat{Q}^r_n(y)}{d\massP}\right).
  \end{equation}
    
 For fixed $y>0$, $v_n(y)$ is increasing in $n$. Since $V_n\leq V$ and $\cM^{\lambda}_a\subseteq\cD^{\lambda}$, we have 
   \begin{equation}  \label{v_nleqv}
     \begin{aligned}  
       v_n(y) =& \inf_{Q\in\cD^{\lambda}}\left\{\massE\left[V_n\left(y\frac{dQ^r}{d\massP}\right)\right]+y\left\langle Q,e_T\right\rangle 
                     + (n+1)y\left(1-\massE\left[\frac{dQ^r}{d\massP}\right]\right)\right\}  \\
              \leq& \inf_{Z_T^0\in\cM^{\lambda}_a}\massE\big[V\left(yZ^0_T\right)+yZ_T^0e_T\big] = v(y),
     \end{aligned}
   \end{equation}
   which means that $v_n$ is dominated by $v$. 
 Therefore, we may define now the function 
   $$ v_{\infty}(y):=\lim_{n\to\infty}v_n(y), \quad y>0,$$
   which turns out later to be the function $v$. 
 In addition, we could prove in by applying \cite[Corollary 4.2]{Sch01} that $v_\infty$ is finitely valued and dominated by $v$, which is as same as \cite[Lemma 2.2]{Owe02}. 
 
 \begin{proof}[Proof of Theorem \ref{maintheoremR}]
  The first step of the proof is to consider the convergence of $v_n(y_n)\rightarrow v_\infty(y)$, as $n\rightarrow \infty$, provided $(y_n)_{n\in \mathbb{N}}$ converging to $y$ in the domain of $v_\infty$. 
  This can be proved by recalling \cite[Lemma 2.3]{Owe02}. 
  In particular, for each $y_n$, denote by $\widehat{Q}_n(y_n)\in\cD^{\lambda}$ the corresponding solution to the  dual problem $v_n(y_n)$. 
  Following the lines of the proof to \cite[Lemma 3.1]{Owe02}, we can show that there exists a measure $\widehat{\massQ}(y)$ such that 
    \begin{equation} \label{convl1}
      \frac{d\widehat{Q}^r_n(y_n)}{d\massP}\xrightarrow{L^1(\massP)} \frac{d\widehat{\massQ}(y)}{d\massP}
    \end{equation}
    and $\Vert \frac{d\widehat Q(y)}{d\massP}\Vert_{L^1(\massP)}=1$. Form \eqref{convl1}, it is clear that 
     $$ \widehat{Q}_n(y_n)\xrightarrow{ba} \widehat{\massQ}(y).$$
  In what follows, we shall show that there exists a couple $\widehat Z(y):=\big(\widehat Z^0(y), \widehat Z^1(y)\big)\in \cZ^\lambda_a(S)$ such that $\widehat Z^0_T(y)=\tfrac{d\widehat{\massQ}(y)}{d\massP}$.   \\
 

  We claim that set $\cD^{\lambda}$ is the $\sigma(ba,L^{\infty})$-closure of $\cM^{\lambda}_a$. 
  Indeed, we only need to prove $\overline{\cM^{\lambda}_a}^{\sigma(ba,L^{\infty})}\supseteq\cD^{\lambda}$. 
  Assume that there exists an element $\widetilde{Q}\in\cD^{\lambda}$ satisfying $\widetilde{Q}\notin\overline{\cM^{\lambda}_a}^{\sigma(ba,L^{\infty})}$. 
  Due to the convexity of $\cM^{\lambda}_a$, its $\sigma(ba,L^{\infty})$-closure $\overline{\cM^{\lambda}_a}^{\sigma(ba,L^{\infty})}$ is also convex. 
  By the Hahn-Banach theorem, there exists $f\in L^{\infty}$, such that $\langle\widetilde{Q},f\rangle>\alpha$ and 
     $$ \langle Q,f\rangle \leq \alpha, \qquad \forall Q\in\overline{\cM^{\lambda}_a}^{\sigma(ba,L^{\infty})}, $$
     for some $\alpha\in\RR$. 
  In particular, $\massE[Z_T^0f]\leq \alpha$ for all $Z^0_T\in\cM_a^{\lambda}$, which follows by Theorem \ref{superreplication} that $f\in\cC^{\lambda}(\alpha)$, therefore $f-\alpha\in\cC^{\lambda}$. 
  By the definition of $\cD^{\lambda}$, we obtain that 
     $$\langle \widetilde{Q},f-\alpha\rangle = \langle \widetilde{Q},f\rangle - \alpha \leq 0, $$
     which contradicts to the fact that $\langle\widetilde{Q},f\rangle>\alpha$.\\
%
     
  Then, by \cite[Proposition A.1]{GLY16a}, there exists a sequence $\big(Z^{n,0}_T(y)\big)_{n\in\NN}\subseteq\cM^\lambda_a$, such that $Z^{n,0}_T(y)$ converges to $\frac{d\widehat\massQ(y)}{d\massP}$ in probability. 
  As $\big\| Z^{n,0}_T(y)\big\Vert_{L^1(\massP)}=\big\Vert\frac{d\widehat\massQ(y)}{d\massP}\big\Vert_{L^1(\massP)}=1$, 
    it follows by Scheff\'e's lemma that $Z^{n,0}_T(y)$ converges to $\frac{d\widehat\massQ(y)}{d\massP}$ in $L^1(\massP)$.
  By Lemma \ref{closednessMa} below, we deduce $\frac{d\widehat\massQ(y)}{d\massP}\in\cM^\lambda_a$. 
  Therefore, we obtain a couple $\widehat Z(y):=\big(\widehat Z^0(y), \widehat Z^1(y)\big)\in \cZ^\lambda_a(S)$ such that $\widehat Z^0_T(y)=\tfrac{d\widehat{\massQ}(y)}{d\massP}$.\\
  
  
  Following the proof of \cite[Corollary 3.2.(i)]{Owe02}, we may show that the map $y\mapsto\widehat{Z}^0_T(y)$ is continuous in the $L^1(\massP)$-norm 
   $$ \lim_{n\to\infty}v_n(y_n) = v(y) =\massE\left[V\big(y\widehat{Z}_T^0(y)\big)+y\widehat{Z}_T^0(y)e_T\right],$$
   and thus $\widehat{Z}_T^0(y)\in\cM^\lambda_a$ is the unique minimizer of the dual problem \eqref{DDP}. 
  The dual value function $v$ is strictly convex. \\

From now on, by following the lines of the proof of \cite[Theorem 1.1]{Owe02}, we may show the other assertions analogously.
 \end{proof}

\begin{lemma}  \label{closednessMa}
  The set $\cM^\lambda_a$ is closed with respect to $L^1(\massP)$-topology.
\end{lemma}

\begin{proof} 
  Consider the sequence $\big(Z^{n,0}_T\big)_{n\in\mathbb N}\subseteq \cM^\lambda_a$, associated with absolutely continuous consistent price systems 
    $Z^n:=(Z^{n,0}_t, Z^{n,1}_t)_{0\leq t\leq T}\in \cZ^\lambda_a(S)$. 
  Moreover, assume that 
    \begin{equation} \label{convl1q}
       Z^{n,0}_T\xrightarrow{L^1(\massP)} Z_T^0, 
    \end{equation}
    for some $Z_T^0\in L^1(\massP)$.
  We now show that $Z_T^0\in \cM^\lambda_a$. 
  Notice that for each $n$, the couple of processes $(Z^{n,0}, Z^{n,1})$ are nonnegative local martingales and thus are supermartingales. 
  By \cite[Theorem 2.7]{CS16a}, there exists a sequence $(\widetilde{Z}^{n,0}, \widetilde{Z}^{n,1})_{n\in \mathbb{N}}$, which is a subsequence of convex combinations of $(Z^{n,0}, Z^{n,1})_{n\in \mathbb{N}}$, 
    i.e., $(\widetilde{Z}^{n,0}, \widetilde{Z}^{n,1})\in \conv\big((Z^{n,0}, Z^{n,1}), (Z^{n+1,0}, Z^{n+1,1}), \cdots\big)$
    and there also exists a couple of nonnegative optional strong supermartingales $(\widehat{Z}^0, \widehat{Z}^1)$ (not necessarily c\`adl\`ag, cf.~\cite[Appendix I]{DM82}) such that for every $[0,T]$-valued stopping time $\sigma$, we have
   \begin{equation}\label{pro}
     \widetilde{Z}^{n,i}_\sigma\stackrel{\massP}{\longrightarrow}  \widehat{Z}^i_\sigma,\quad {\rm as}\ n\rightarrow \infty,\quad i=0, 1.
   \end{equation}
  Since the set $\cZ^\lambda_a(S)$ is closed under countable convex combinations (cf.~\cite[Lemma A.1]{CS15}), 
    we have for each $n\in\NN$, $(\widetilde{Z}^{n,0}, \widetilde{Z}^{n,1})\in \cZ^\lambda_a(S)$ and particularly, $\widetilde Z^{n,0}$ is a martingale. 
  From \eqref{convl1q}, $\widetilde{Z}^{n,0}_T$ converges to $\widehat Z^0_T$ in $L^1(\massP)$. 
  Therefore, we can prove that $\widehat{Z}^0$ is a martingale by Fatou's Lemma and obviously, $\widehat{Z}^0_T=Z^0_T$. 
  Moreover, for every $[0, T]$-valued stopping time $\tau$, the sequence $\big(\widetilde{Z}^{n,0}_\tau\big)_{n\in \mathbb{N}}$ is uniformly integrable.\\

  Since $S$ is locally bounded and $(\widetilde{Z}^{n,1})_{n\in \mathbb{N}}$ is a sequence of local martingales, one can choose a sequence of stopping times $(\tau_m)_{m\in\mathbb{N}}$ increasing and converges almost surely to $\infty$, 
     such that each stopped process $S^{\tau_m}$ is bounded and, for each $n$, $\widetilde{Z}^{n,1}_{\cdot\wedge\tau_m}$ is a martingale. 
  From $(\ref{J10})$, we have $\widetilde{Z}^{n,1}\leq \widetilde{Z}^{n,0} S$. 
  Thus, for each $m$, $\big(\widetilde{Z}^{n,1}_{\tau_m}\big)_{n\in \mathbb{N}}$ is uniformly integrable, which implies that $\widetilde{Z}^{n,1}_{\tau_m}$ converges to $\widehat{Z}^1_{\tau_m}$ in $L^1(\massP)$ 
     and $\widetilde{Z}^{n,1}_{\cdot\wedge\tau_m}$ is a martingale. 
  Consequently, we conclude that $(\widehat{Z}^0, \widehat{Z}^1)\in \cZ^\lambda_a(S)$ by observing that \eqref{J10} is satisfied by $(\widehat{Z}^0, \widehat{Z}^1)$.
 \end{proof}

 We now consider the question, whether there exists a self-financing trading strategy $(\widehat{\varphi}^0,\widehat{\varphi}^1)$ under transaction costs $\lambda$, that attains the solution $\widehat{g}(x)$ to \eqref{PPP}, 
   i.e., $V^{\textnormal{liq}}_T\big(\widehat{\varphi}\big)=\widehat{g}(x)$. 
 As in \cite{CS16b}, we define as follows the set $\cA^{\lambda}_U(x)$ of all attainable trading strategy.  We simply note $\cA^\lambda_U$ by $\cA^\lambda_U(0)$. 
 \begin{definition}
 We call $\varphi:=(\varphi^0, \varphi^1)$ attainable trading strategy, if $\varphi$ is a  predictable and of  finite variation, starting at $(\varphi^0_0,\varphi^1_0)=(x,0)$, satisfying the $\lambda$-self-financing condition \eqref{SF} 
   and such that there exists a sequence $(\varphi^{n,0},\varphi^{n,1})_{n\in\NN}\subseteq \cA^{\lambda}_{adm}(x)$ varifying that $U\big(V^{\textnormal{liq}}_T(\varphi^n)+e_T\big)\in L^1(\massP)$, 
    $$ U\big(V^{\textnormal{liq}}_T(\varphi^n)+e_T\big)\xrightarrow{\,\,L^1(\massP)\,\,}U\big(V^{\textnormal{liq}}_T(\varphi)+e_T\big) $$ 
    and 
    $$ \massP\left[\big(\varphi^{n,0}_t, \varphi^{n,1}_t\big)\rightarrow\big(\varphi^0_t, \varphi^1_t\big),\,\,\forall t\in[0,T]\right] = 1. $$
 \end{definition}
 
 \begin{proposition}  \label{finiteEntropyProp}
   In addition to the assumptions of Theorem \ref{maintheoremR}, we assume furthermore that, for some $\lambda'\in(0,\lambda)$, there exists a $\lambda'$-consistent price system $\big(\overline{Z}^0,\overline{Z}^1\big)\in\cZ^{\lambda'}_e(S)$, 
    such that 
      $$ \massE\big[V\big(\overline{y}\overline{Z}^0_T\big)\big]<\infty, $$
    for some $\overline{y}>0$.  Then the solution to the primal problem \eqref{PPP_Cu} is attainable, i.e., there exists a $\big(\widehat{\varphi}^0,\widehat{\varphi}^1\big)\in\cA^\lambda_U$ 
    such that $V^{\textnormal{liq}}_T\big(\widehat{\varphi}\big) = \widehat{g}(x)$, 
    and the dual optimizer $(\widehat{Z}^0,\widehat{Z}^1)$ belongs to $\cZ^{\lambda}_e(S)$, i.e.,  $(\widehat{Z}^0,\widehat{Z}^1)$ is a $\lambda$-consistent price system. 
 \end{proposition}

 \begin{proof}  
    By Theorem \ref{maintheoremR}, there exists a sequence $\big((\varphi^{n,0},\varphi^{n,1})\big)_{n\in\NN}\subseteq\cA^\lambda_{adm}$ such that 
     \begin{equation}  \label{U(x+Vn+e)toU(x+g+e)}
       U\big(x+V^{\textnormal{liq}}_T(\varphi^n)+e_T\big)\xrightarrow{\,\,L^1(\massP)\,\,}U\big(x+\widehat{g}(x)+e_T\big). 
     \end{equation}
    Then, for $\overline{S}:=\frac{\overline{Z}^1}{\overline{Z}^0}$, the process $\big(\overline{Z}^0_t(x+\varphi_t^{n,0}+\varphi_t^{n,1}\overline{S}_t + A_t^n)\big)_{0\leq t\leq T}$ is a supermartingale for each $n\in\NN$, 
      where $$ A_t^n := (\lambda-\lambda')\int_0^tS_ud\varphi_u^{n,1,\downarrow}. $$
    Indeed, by integration by parts we may write 
      $$ x+\varphi_t^{n,0}+\varphi_t^{n,1}\overline{S}_t + A_t^n = x+\varphi_t^{n,0}+ \int_0^t\varphi_u^{n,1}d\overline{S}_u +\int_0^t\overline{S}_ud\varphi_u^{n,1}+ A_t^n. $$
    Since $(1-\lambda)S_u\leq \overline{S}_u\leq S_u$ and by the $\lambda$-self-financing condition \eqref{SF}, we obtain 
     \begin{equation*}
       \begin{aligned}
         \varphi_t^{n,0} &-\varphi_s^{n,0} +\int_s^t\overline{S}_ud\varphi_u^{n,1}+ A_t^n-A_s^n \\
           & \leq  -\int_s^tS_ud\varphi_u^{n,1,\uparrow} + \int_s^t(1-\lambda)S_ud\varphi_u^{n,1,\downarrow} +\int_s^t\overline{S}_ud\varphi_u^{n,1}+ \int_s^t(\lambda-\lambda')S_ud\varphi_u^{n,1,\downarrow} \\
           & = -\int_s^t\big(S_u-\overline{S}_u\big)d\varphi_u^{n,1,\uparrow} - \int_s^t\big(\overline{S}_u-(1-\lambda')S_u\big)d\varphi_u^{n,1,\downarrow} \leq 0, \\
       \end{aligned}
     \end{equation*}
      for all $0\leq s<t\leq T$, therefore the process $(B^n_t)_{0\leq t\leq T}:= \left(\varphi_t^{n,0}+ \int_0^t\overline{S}_ud\varphi_u^{n,1}+ A_t^n\right)_{0\leq t\leq T}$ is nonincreasing. 
    It follows by Bayes' rule that $\overline{S}$ is a local martingale under the measure $\overline{\massQ}\sim\massP$ defined by $\frac{d\overline{\massQ}}{d\massP}:=\overline{Z}^0_T$. 
    As $\varphi^{n,1}$ is of finite variation and hence locally bounded, the stochastic integral $\varphi^{n,1}\sint\overline{S}$ is a local martingale under $\overline{\massQ}$. 
    Therefore, $x+\big(\varphi^{n,1}\sint\overline{S}\big)_t + B_t^n$ is a local supermartingale under $\overline{\massQ}$. 
    Using Bayes' rule once again, we obtain that 
       $$\big(\overline{Z}^0_t\big(x + \varphi_t^{n,0}+\varphi_t^{n,1}\overline{S}_t + A_t^n \big)\big)_{0\leq t\leq T} = \big(\overline{Z}^0_t\big(x + \big(\varphi^{n,1}\sint\overline{S}\big)_t+B_t^n\big)\big)_{0\leq t\leq T}$$ 
     is a local supermartingale under $\massP$. 
    Since $(\varphi^{n,0},\varphi^{n,1})\in\cA^\lambda_{adm}$, we have 
       $$ \overline{Z}^0_t\big(x + \varphi_t^{n,0}+\varphi_t^{n,1}\overline{S}_t + A_t^n \big) \geq \overline{Z}^0_tV^{\textnormal{liq}}_t(\varphi^n) \geq - M^n\overline{Z}^0_t, $$
      for some $M^n\geq 0$. 
    As $\overline{Z}^0$ is a true martingale, the process $\big(\overline{Z}^0_t\big(x + \varphi_t^{n,0}+\varphi_t^{n,1}\overline{S}_t + A_t^n \big)\big)_{0\leq t\leq T}$ is a true supermartingale under $\massP$, 
      which implies in particular that 
       $$ \massE\big[\overline{Z}^0_T\big(x+\varphi_T^{n,0}+A_T^n\big)\big] \leq x, $$
      and 
       \begin{equation}  \label{E[Z(x+phi+e+A)]leqx}
         \massE\big[\overline{Z}^0_T\big(x+\varphi_T^{n,0}+A_T^n+e_T\big)\big] \leq x + \rho,
       \end{equation}
       for all $n\in\NN$. \\

    By Fenchel's inequality and the monotonicity of $U$ we can estimate 
     \begin{equation*}
       \overline{y}\overline{Z}^0_T\big(x+V^{\textnormal{liq}}_T(\varphi^n)+A_T^n+e_T\big) \geq U\big(x+V^{\textnormal{liq}}_T(\varphi^n)+e_T\big)-V\big(\overline{y}\overline{Z}^0_T\big). 
     \end{equation*}
    By the assumption we have that $V\big(\overline{y}\overline{Z}^0_T\big)\in L^1(\massP)$, and it follows by \eqref{U(x+Vn+e)toU(x+g+e)} that
      $$\left(\overline{y}\overline{Z}^0_T\big(x+V^{\textnormal{liq}}_T(\varphi^n)+A_T^n+e_T\big)^-\right)_{n\in\NN} \mbox{ is uniformly integrable.}$$  
    Together with \eqref{E[Z(x+phi+e+A)]leqx} we obtain the sequence $\left(\overline{y}\overline{Z}^0_T\big(x+V^{\textnormal{liq}}_T(\varphi^n)+A_T^n+e_T\big)\right)_{n\in\NN}$ is $L^1(\massP)$-bounded.  \\
      
    It follows from $\overline{Z}^0_T >0$ and $V^{\textnormal{liq}}_T(\varphi^n)\xrightarrow{\,\massP\,}\widehat{g}(x)$, that $\conv\{A_T^n;\,n\in\NN\}$ is bounded in $L^0(\massP)$.
    Since $\overline{S}$ is a nonnegative local martingale under $\overline{\massQ}$, hence also a nonnegative supermartingale under $\overline{\massQ}$, we see that
      $$ \inf_{0\leq u\leq T}S_u \geq \inf_{0\leq u\leq T}\overline{S}_u  > 0 $$
     by \cite[Theorem VI-17]{DM82}. 
    This implies that $\conv\{\Var_T(\varphi^{n,1});\,n\in\NN\}$ is bounded in $L^0(\massP)$, therefore the same for $\conv\{\Var_T(\varphi^{n,0});\,n\in\NN\}$. 
    By \cite[Proposition 3.4]{CS06} there exists a sequence 
      $$ \big(\widetilde{\varphi}^{n,0},\widetilde{\varphi}^{n,1}\big)\in\conv\big\{\big(\varphi^{k,0},\varphi^{k,1}\big);\, k\geq n\big\} $$
      of convex combinations and a predictable process $\big(\widehat{\varphi}^0,\widehat{\varphi}^1\big)$ of finite variation such that 
        $$ \massP\left[\big(\widetilde{\varphi}^{n,0}_t, \widetilde{\varphi}^{n,1}_t\big)\rightarrow\big(\widehat{\varphi}^0_t, \widehat{\varphi}^1_t\big),\,\,\forall t\in[0,T]\right] = 1. $$
    It implies that $(\widehat{\varphi}^0,\widehat{\varphi}^1)$ is a $\lambda$-self-financing trading strategy such that $V^{\textnormal{liq}}_T\big(\widehat{\varphi}\big) = \widehat{g}(x)$, 
      therefore $(\widehat{\varphi}^0,\widehat{\varphi}^1)\in\cA^\lambda_U$. \\
     
    Since $\widehat{g}(x)=V^{\textnormal{liq}}_T(\widehat{\varphi})<\infty$, we have that 
      $$ \widehat{y}\widehat{Z}^0_T(\widehat{y}) = U'\big(x+\widehat{g}(x)+e_T\big)>0 $$
      by the Inada condition. 
    Hence, $(\widehat{Z}^0,\widehat{Z}^0)\in\cZ^{\lambda}_e(S)$. 
  \end{proof}

\section{Shadow price}
 For utility maximization problems with proportional transaction costs, it has been observed that the original market with transaction costs can sometimes be replaced by a frictionless shadow market, 
   that yields the same optimal strategy and utility.  
 In this section, we shall study the existence of such a fictitious market for the utility maximization considered in the previous section. 
 In particular, we first construct shadow prices in the classical sense via the dual optimizer of the problem \eqref{DDP}. Then, we introduce a generalized definition of shadow price processes and show that this kind of process always exists 
   whenever the utility maximization problem can be solved via the duality approach. \\
   
 
 First, we adapt the definition of shadow price processes \cite[Definition 2.2]{CS16b} in the classical sense to our setting with random endowment. 
 
 \begin{definition}  \label{classicalshadowprice}
   A semimartingale $\widetilde{S}=\big(\widetilde{S}_t\big)_{0\leq t\leq T}$ is called a shadow price, if 
    \begin{enumerate}[$(i)$]
     \item $\widetilde{S}$ takes values in the bid-ask spread $[(1-\lambda)S,S]$;
     \item The solution $\widehat{H}=\big(\widehat{H}_t\big)_{0\leq t\leq T}$ to the frictionless utility maximization problem 
            \begin{equation}  \label{frictionlessPPAu}
              u\big(x;\widetilde{S}\big) := \sup_{H\in\cA_U\big(\widetilde{S}\big)}\massE\big[U\big(x+\big(H\sint\widetilde{S}\big)_T+e_T\big)\big],
            \end{equation}
            exists in the sense of \cite{Owe02}, where $\cA_U\big(\widetilde{S}\big)$ denotes the set of all $\widetilde{S}$-integrable predictable processes $H$, 
            such that there exists a sequence $(H^n)_{n\in\NN}$ of admissible self-financing trading strategies without transaction costs such that $U\big(x+(H^n\sint\widetilde{S})_T+e_T\big)\in L^1(\massP)$ and 
              $$ U\big(x+(H^n\sint\widetilde{S})_T+e_T\big) \xrightarrow{\,\,L^1(\massP)\,\,}U\big(x+(H\sint\widetilde{S})_T+e_T\big);$$
     \item The optimal trading strategy $\widehat{H}$ to the frictionless problem \eqref{frictionlessPPAu} coincides with the holdings in stocks $\widehat{\varphi}^1$ to the utility maximization problem \eqref{PPP_Cu} 
             under transaction costs such that $\big(\widehat{H}\sint\widetilde{S}\big)_T = \widehat{g}(x)= V^{\textnormal{liq}}_T(\widehat{\varphi})$.
    \end{enumerate}
 \end{definition}

 Generally speaking, a classical shadow price $\widetilde{S}$ (if exists) allows us to obtain the optimal trading strategy for the utility maximization problem \eqref{PPP_Cu} 
   by solving the frictionless utility maximization problem \eqref{classicalshadowprice}. 
 Obviously, trading for $\widetilde{S}$ yields higher expected utility than trading under transaction costs. 
 Thus, the shadow price is a least favorable frictionless market lying in the bid-ask spread. 
 It also follows that the optimal strategy $\widehat{H}$ to the problem \eqref{frictionlessPPAu} in the shadow market is of finite variation, due to the coincidence of $\widehat H$ and $\widehat \varphi^1$. 
 Furthermore, both of them only trade if $\widehat{S}$ is at the bid or ask price, i.e., 
     $$\big\{d\widehat{\varphi}^1>0\big\}\subseteq \big\{\widehat{S}=S\big\} \quad\mbox{ and }\quad \big\{d\widehat{\varphi}^1<0\big\}\subseteq \big\{\widehat{S}=(1-\lambda)S\big\}.$$
     We refer the readers to \cite{CS15} for details of the above notion. \\
     
    
 The following proposition provides the sufficient condition for the existence of a shadow price, in particular, this is guaranteed by the existence of a $\lambda'$-consistent price system with finite $V$-expectation, where $0<\lambda'<\lambda$.
 
 \begin{proposition}\label{excla}
   Under the assumptions of Theorem \ref{maintheoremR}, suppose that the solution $\widehat{g}(x)$ to the primal problem \eqref{PPP_Cu} is attainable and 
     the solution $\widehat Z:= \big(\widehat{Z}^0,\widehat{Z}^1\big)$ to the dual problem \eqref{DDP} belongs to $\cZ^{\lambda}_e(S)$, i.e., $\widehat Z$ is a $\lambda$-consistent price system. 
   Then, the process defined by $\widehat{S}:=\frac{\widehat{Z}^1}{\widehat{Z}^0}$ is a shadow price to problem \ref{PPP_Cu} in the sense of Definition \ref{classicalshadowprice}. 
 \end{proposition}

 \begin{proof}
   By the assumption that $\widehat \varphi$ is attainable, we know that  there exists a sequence of admissible $\lambda$-self-financing trading strategies $\big(\varphi^{n,0},\varphi^{n,1}\big)_{n\in\NN}$ satisfying
     \begin{equation}\label{convsen}
       \massP\left[\big(\varphi^{n,0}_t,\varphi^{n,1}_t\big)\longrightarrow\big(\widehat{\varphi}^0_t,\widehat{\varphi}^1_t\big),\,\forall t\in[0,T]\right] = 1,  
     \end{equation}
     and 
     \begin{equation} \label{U(x+phin+e)L1toU(x+phi+e)}
       U\big(x+\varphi^{n,0}_T+e_T\big) \xrightarrow{\,\,L^1(\massP)\,\,} U\big(x+\widehat{g}(x)+e_T\big). 
     \end{equation}
   By Frechel's conjugate, and the fact that $\massE[\widehat{Z}^0_T\big(x+\widehat{g}(x)+e_T\big)]<x+\rho$, we could show in a similar way as Proposition \ref{finiteEntropyProp} that 
        $$\left(\widehat{y}\widehat{Z}^0_T\left(x+\varphi^{n,0}_T +e_T\right)^-\right)_{n\in \mathbb{N}}$$ 
      is uniformly integrable. 
   Due to the boundedness of $e_T$, we have $\left(\widehat{y}\widehat{Z}^0_T\left(x+\varphi^{n,0}_T\right)^-\right)_{n\in \mathbb{N}}$ is also uniformly integrable.
   Then, the proof goes in a same way as the one for \cite[Proposition 3.3]{CS16b}. 
   In particular, we have $\big(\widehat{Z}^0_t(x+\widehat{\varphi}^{0}_t)+\widehat{Z}_t^1\widehat{\varphi}_t^{1}\big)_{0\leq t\leq T}$ is a supermartingale under $\massP$, 
      since it is the limit of a sequence of supermartingles in the sense \eqref{convsen}. 
   Thanks to Theorem \ref{maintheoremR} (4), we have moreover
       $$ x = \massE\left[\widehat{Z}^0_T\left(x+\widehat{\varphi}^0_T\right)\right], $$
      from which we could conclude that $\big(\widehat{Z}^0_t(x+\widehat{\varphi}^{0}_t)+\widehat{Z}_t^1\widehat{\varphi}_t^{1}\big)_{0\leq t\leq T}$ is a martingale under $\massP$. \\
    
   
   We apply integration by parts to obtain that
     \begin{equation*}
       \begin{aligned}
         \widehat{Z}^0_t\left(x+\widehat{\varphi}^{0}_t\right) &+\widehat{Z}_t^1\widehat{\varphi}_t^{1} = \widehat{Z}^0_t\left(x+\widehat{\varphi}^{0}_t+\widehat{\varphi}_t^{1}\widehat{S}_t\right) \\
           &= \widehat{Z}^0_t\left(x + \big(\widehat{\varphi}^1\sint\widehat{S}\big)_t - \int_0^t\big(S_u-\widehat{S}_u\big)d\widehat{\varphi}^{1,\uparrow}_u - 
                      \int_0^t\big(\widehat{S}_u-(1-\lambda)S_u\big)d\widehat{\varphi}^{1,\downarrow}_u\right) \\
           &=: \widehat{Z}^0_t\left(x+\big(\widehat{\varphi}^1\sint\widehat{S}\big)_t-A_t\right).
       \end{aligned}
     \end{equation*}
   Again as in the proof of \cite[Proposition 3.3]{CS16b}, one could prove that $\widehat{Z}^0\big(x+\widehat{\varphi}^1\sint\widehat{S}\big)$ is a local martingale 
     and $A$ is increasing, which implies that $A\equiv 0$ and furthermore 
       $$\big\{d\widehat{\varphi}^1>0\big\}\subseteq \big\{\widehat{S}=S\big\} \quad\mbox{ and }\quad \big\{d\widehat{\varphi}^1<0\big\}\subseteq \big\{\widehat{S}=(1-\lambda)S\big\}.$$
    
   It is clear that, 
      $$ u\big(x;\widehat{S}\big) \leq \massE\left[V(yZ_T)+yZ_Te_T\right]+xy, $$
     for $y>0$ and $Z_T\in\cZ_a\big(\widehat{S}\big)$.    
   As $\big(\widehat{Z}^0,\widehat{Z}^1\big)\in\cZ^\lambda_e(S)$, we obtain that $\widehat{Z}^0$ is the density process of an equivalent local martingale measure for the frictionless process $\widehat{S}$, 
     therefore 
      $$ \massE\left[V\big(\widehat{y}\widehat{Z}^0_T\big)+\widehat{y}\widehat{Z}^0_Te_T\right] + x\widehat{y} = u(x)\leq u\big(x;\widehat{S}\big)
              \leq \massE\left[V\big(\widehat{y}\widehat{Z}^0_T\big)+\widehat{y}\widehat{Z}^0_Te_T\right] + x\widehat{y}. $$
   It follows from the frictionless duality theorem \cite[Theorem 1.1]{Owe02} that $\widehat{y}$ and $\widehat{Z}^0$ is the optimizer of the frictionless dual problem and moreover, 
       $x+(\widehat{\varphi}^1\sint\widehat{S})_T+e_T=x+\widehat{g}(x)+e_T=-V'\big(\widehat{y}\widehat{Z}^0_T\big)$ 
     is the optimal terminal wealth to the frictionless problem (\ref{frictionlessPPAu}) under $\widehat{S}$. 
   Since $\widehat{\varphi}\sint\widehat{S}$ is a martingale under $\widehat{\massQ}$ given by $\frac{d\widehat{\massQ}}{d\massP}:=\widehat{Z}^0_T$, 
     we obtain that $\widehat{\varphi}^1$ has to be the optimal strategy and in $\cA_U\big(x;\widehat{S}\big)$ by \cite[Theorem 1.1.(v)]{Owe02}. \\
     
              
   In conclusion, the price process $\widehat{S}$ is a shadow price in the sense of Definition \ref{classicalshadowprice} for the utility maximization problem \eqref{PPP_Cu} under transaction costs. 
 \end{proof}
 

 As explained in \cite{CS16b}, that $U(\infty)=\infty$ is a sufficient condition for the existence of classical shadow price by Proposition \ref{excla}, 
   which implies $\widehat{Z}\in\mathcal{Z}^\lambda_e(S)$ and thus the result of Proposition \ref{finiteEntropyProp} holds. 
 When $U(\infty)<\infty$, we observe that the solution to the primal problem \eqref{PPP_Cu} is not necessarily attainable, 
   i.e., there may not exist an optimal $\lambda$-self-financing trading strategy $(\widehat{\varphi}^0,\widehat{\varphi}^1)$, 
   such that $\widehat{\varphi}^0_T=\widehat{g}(x)$. 
 However, the solution $(\widehat{Z}^0,\widehat{Z}^1)$ to the dual problem is always a local martingale (an absolutely continuous consistent price system). 
 We may define the following generalized shadow price, which only leads to the same optimal utility as the one under transaction costs. 

 \begin{definition}
   We keep all settings of Theorem \ref{maintheoremR}.  A semimartingale $\widetilde{S}=(\widetilde{S}_t)_{0\leq t\leq T}$ is called generalized shadow price for the optimization problem \eqref{PPP}, 
    if 
    \begin{enumerate}[$(i)$]
     \item The process $\widetilde{S}$ takes values in  $[(1-\lambda)S,S]$.
     \item The solution $\widetilde{g}\in\cC_U(\widetilde{S})$ to the corresponding frictionless utility maximization problem 
            \begin{equation} \label{Pshadowproblem}
              u(x;\widetilde{S}):= \sup_{g\in\cC_U(\widetilde{S})}\massE[U(x+g+e_T)]
            \end{equation}
            exists and coincides with the optimal solution $\widehat{g}\in\cC^{\lambda}_U$ to \eqref{PPP} under transaction costs, 
            where         
           \begin{eqnarray*}
             \cC_{U}(\widetilde{S}):= \left\{g\in L^0(\massP; \RR\cup\{\infty\})\, \Bigg| \,  \begin{array}{r}
                                                                                 \exists g_n\in\cC(\widetilde{S})\,\mbox{ s.t. }\, U(x+g_n+e_T)\in L^1(\massP) \mbox{ and } \\
                                                                                 U(x+g_n+e_T) \xrightarrow{L^1(\massP)} U(x+g+e_T)
                                                                              \end{array}
                        \right\}, 
           \end{eqnarray*}
           and 
             $$ \cC(\widetilde{S}):=\{g\in L^0\, | \,g\leq (H\sint S)_T\,\mbox{ for some admissible portfolio}\, H\}.$$
    \end{enumerate}
 \end{definition}

\begin{remark}
   In the duality theorem of the utility maximization problem with utility functions defined on the positive real line, 
     the existence of an optimal trading strategy $\widehat{\varphi}\in\cA^\lambda_{adm}$ follows directly from the existence of the dual optimizer $\widehat{g}\in\cC^\lambda$. 
   Therefore, it is quite natural to require in the classical definition of shadow price that the optimal trading strategy in the frictionless shadow market is also an optimal one in the original market with transaction costs. 
   For the problem with utility functions defined on the whole real line, we have in general no chance to find the optimal strategy. 
   This is our motivation to define the generalized shadow price in such a way. 
 \end{remark}
 
 \begin{theorem}
   The semimartingale $\widehat{S}$ defined by \eqref{J10} associated with the solution $(\widehat{Z}^0, \widehat{Z}^1)\in \cZ_a^{\lambda}(S)$ of the dual problem \eqref{DDP} is a generalized shadow price by the definition above.
 \end{theorem}

 \begin{proof}
  From the definition of $\cC_U(\widehat{S})$ and $\cC(\widehat{S})$,  we know 
    $$ u(x;\widehat{S}) = \sup_{g\in\cC(\widehat{S})}\massE[U(x+g+e_T)]. $$
  Since  $\cC^{\lambda}\subseteq\cC(\widehat{S})$, then
     \begin{equation} \label{ulequ(S)}
      u(x) = \sup_{g\in\cC^{\lambda}}\massE[U(x+g+e_T)] \leq \sup_{g\in\cC(\widehat{S})}\massE[U(x+g+e_T)] = u(x;\widehat{S}).
    \end{equation}
    Moreover,
    \begin{equation} \label{D(S)subsetDlambda}
     \begin{aligned}
       \cD(\widehat{S}):=& \,\big\{Q\in ba\, \big|\, \|Q\|=1\,\mbox{ and }\, \langle Q,g \rangle\leq 0 \mbox{ for all } g\in\cC(\widehat{S})\cap L^{\infty}\big\}   \\
                   \subseteq & \,\big\{Q\in ba\, \big|\, \|Q\|=1\,\mbox{ and }\, \langle Q,g \rangle\leq 0 \mbox{ for all } g\in\cC^{\lambda}\cap L^{\infty}\big\} = \cD^{\lambda}.  \\
     \end{aligned}
    \end{equation}
  Let $\widehat{y}:=u'(x)$. Now consider the following value function
    \begin{equation*} \label{Dshadowproblem}
        v(\widehat{y};\widehat{S}) := \inf_{\massQ\in\cM_a(\widehat{S})}\massE\left[V\left(\widehat{y}\frac{d\massQ}{d\massP}\right)+\widehat{y}\frac{d\massQ}{d\massP}e_T\right].
    \end{equation*}
  By \cite[Corollary A.2.]{GLY16a}, the formulation of the function $v(\cdot)$ is equivalent to 
    \begin{equation}\label{generaldual}
        v(\widehat{y};\widehat{S}) = \inf_{Q\in\cD(\widehat{S})}\left\{\massE\left[V\left(\widehat{y}\frac{dQ^r}{d\massP}\right)\right]+\widehat{y}\langle Q,e_T\rangle\right\}. 
    \end{equation}       
  Then, we deduce from \eqref{D(S)subsetDlambda}, 
    \begin{equation}  \label{v(y,S)geqv(y)}
      \begin{aligned}
        v(\widehat{y};\widehat{S}) &\geq \inf_{Q\in\cD^{\lambda}}\left\{\massE\left[V\left(\widehat{y}\frac{dQ^r}{d\massP}\right)\right]+\widehat{y}\langle Q,e_T\rangle\right\} = v(\widehat{y}).
      \end{aligned}
    \end{equation}
  As $(\widehat{Z}^0,\widehat{Z}^1)\in\cZ^\lambda_a(S)$, we have that the measure $\widehat{\massQ}$, defined by $\frac{d\widehat{\massQ}}{d\massP}=\widehat{Z}_T^0$, 
    is an absolutely continuous martingale measure for $\widehat{S}$, i.e., $\widehat{\massQ}\in\cM_a(\widehat{S})$. 
  Hence, we deduce that $\widehat{Z}^0_T$ a fortiori is the optimizer for $v(\widehat{y};\widehat{S})$. 
  In particular,  $v(\widehat{y}) = v(\widehat{y};\widehat{S})$.  
  It follows from Theorem \ref{maintheoremR}, Fenchel's inequality and \eqref{ulequ(S)} that  
    \begin{equation}\label{wangbaching}
    u(x) = v(\widehat{y}) + x\widehat{y} = v(\widehat{y};\widehat{S})+ x\widehat{y}\geq \inf_{y>0}\left\{v\big(y;\widehat{S}\big)+xy\right\} \geq  u(x;\widehat{S})\geq  u(x), 
    \end{equation}
    therefore the primal value functions coincide. 
  In the frictionless market, we have a posteriori $u(x; \widehat{S})<U(\infty)$.
  By the uniqueness of the primal solution and $\cC^{\lambda}\subseteq\cC(\widehat{S})$, 
    the primal optimizer to \eqref{Pshadowproblem} exists, is unique and coincides with the one to the optimization problem \eqref{PPP}.
 \end{proof}
 
\begin{remark}  
  In the theorem above, it is not clear whether  equivalent martingale measures for the shadow market $\widehat{S}$ exist or not, except for the case where $(\widehat{Z}^0,\widehat{Z}^1)$ is strictly positive. 
  We stress that the following inequity in \eqref{wangbaching} still holds true under the assumption $\cM_a(\widehat{S})\neq\emptyset$: 
    $$ u(x;\widehat{S})\leq \inf_{y>0}\left\{v\big(y;\widehat{S}\big)+xy\right\}. $$
  Indeed, this follows from Fenchel's inequality and the easy part of the superreplication theorem in the frictionless setting, which could be deduced under the weaker assumption $\cM_a(\widehat{S})\neq\emptyset$. 
  Furthermore, we observe that there is no duality gap, i.e., 
   $$ u\big(x;\widehat{S}\big) = v\big(\widehat{y};\widehat{S}\big) + x\widehat{y} = \inf_{y>0}\left\{v\big(y;\widehat{S}\big)+xy\right\}, $$
   and there exist at least a primal optimizer (which may not be attained by trading strategies) and a dual one in the shadow market, which coincide with the ones in the original market with transaction costs. 
\end{remark}

\begin{remark}
 The fact that $\widehat{Z}^0_T\in\cM^{\lambda}_a$ (or $\cM^{\lambda}_e$) is the unique solution to the dual problem \eqref{DDP} does not mean the uniqueness of the couple 
     $(\widehat{Z}^1,\widehat{Z}^1)\in\cZ^{\lambda}_a(S)$ (or $\cZ^{\lambda}_e(S)$). 
 In another word, the shadow price process need not be unique. 
\end{remark}

 Conversely, the following result shows that, if a (generalized) shadow price $\widehat{S}$ exists as above and satisfies $\cM_e(\widehat{S})\neq\emptyset$, it is necessarily derived from a dual minimizer. 
   (Compare \cite[Proposition 3.8]{CS15}.)
 
 \begin{proposition}
  If a (generalized) shadow price $\widehat{S}$ exists as above and satisfies $\cM_e(\widehat{S})\neq\emptyset$, 
     then there exists a $\massP$-martingale $\widehat Z^0$, such that 
     $(\widehat{Z}^0,\widehat{Z}^0\widehat S)\in\cZ^{\lambda}_a(S)$ is a solution to the dual problem \eqref{DDP}. 
 \end{proposition}

 \begin{proof}
  Choose  $\massQ\in\cM_a\big(\widehat{S}\big)$ and denote by $Z$ its density process. 
  It is obvious that $(Z^0,Z^1):=\big(Z,Z\widehat{S}\big)\in\cZ^\lambda_a(\widehat{S})$.
  Moreover, from $\cM_a\big(\widehat{S}\big)\subseteq \mathcal M^\lambda_a$ and \cite[Theorem 2.2]{Sch01}, we have
      \begin{align*}
        u(x) &= v\big(\widehat y(x)\big)+x\widehat y(x) \leq v\big(\widehat y(x;\widehat{S})\big)+x\widehat y(x;\widehat{S}) \\
             &\leq v\big(\widehat y(x;\widehat{S});\widehat{S}\big)+x\widehat y(x;\widehat{S})=u(x;\widehat{S})=u(x),  
      \end{align*}
      which implies $\widehat y(x)=\widehat y(x;\widehat{S})$ and $v\big(\widehat y(x)\big)=v\big(\widehat y(x;\widehat{S});\widehat{S}\big)$, 
       hence $\big(\widehat{Z}^0,\widehat{Z}^1\big):=\big(\widehat{Z},\widehat{Z}\widehat{S}\big)\in\cZ^\lambda_a(S)$ is the solution to the frictional dual problem \eqref{DDP}, 
      where $\widehat{Z}\in\cM_a\big(\widehat{S}\big)$ is the solution to its frictionless counterpart for the shadow price process $\widehat{S}$. 
 \end{proof}
 
 \begin{remark}
   The assumption $\cM_e(\widehat{S})\neq\emptyset$ ensures that we could apply the result of  \cite[Theorem 2.2]{Sch01} to the frictionless market with $\widehat{S}$, in particular, we could deduce the following equality
     $$ v\big(\widehat y(x;\widehat{S});\widehat{S}\big)+x\widehat y(x;\widehat{S})=u(x;\widehat{S}).$$
 \end{remark}

%
%

\vspace{3mm}

It is known from the so-called ``face-lifting theorem'' that, under transaction costs, the bounds for option prices obtained from superreplication arguments are only the trivial bounds. (See e.g., \cite{GRS08}.)
Therefore, the concepts of superreplication do not make sense economically in the presence of transaction costs. 
However, the concept of a utility indifference price makes perfect economic sense in the presence of transaction costs. (See e.g., \cite{HN89}.)\\

We denote now the value function by $u^{e_T}(x)$ instead of $u(x)$ to emphasize the dependence on $e_T$ and $u^0$ denotes the value function of utility maximization problem without random endowment.
The utility indifference price is the solution $p(x)$ of 
  $$ u^{e_T}\big(x-p(x)\big)=u^0(x). $$

Let us consider the exponential utility function 
  $$ U(x)=-\exp(-\gamma x), \quad x\in\RR, $$
  where $\gamma>0$ stands for the absolute risk aversion parameter. 
In this case, using the duality result, we could obtain a dual formulation for the utility based price. \\

For the exponential utility function $U(x)$, we have 
  $$ V(y)=\frac{y}{\gamma}\left(\log\left(\frac{y}{\gamma}\right)-1\right), \quad y >0.  $$
  
\begin{lemma}  \label{u(x)=infE}
  For the exponential utility function, we have that
    \begin{equation*}
      u^{e_T}(x) = \inf_{Z_T^0\in\cM^\lambda_a}U\left(\frac{1}{\gamma}\massE\left[Z_T^0\log\big(Z_T^0\big)\right]+\massE\left[Z_T^0e_T\right] +x \right),
    \end{equation*}
    for all $x\in\RR$.
\end{lemma}
The proof of the above lemma follows from Theorem \ref{maintheoremR} and is similar to the one of \cite[Proposition 11]{Bou02}, so we omit it. We now introduce the utility based pricing for $e_T$.
\begin{proposition}
  For all $x\in\RR$, the utility based price of $e_T$ equals
       \begin{equation*}
      \begin{aligned}
         p(x) &= U^{-1}\big(u^{e_T}(x)\big) - U^{-1}\big(u^0(x)\big)  \\
              &= \inf_{Z^0_T\in\cM^\lambda_a}\massE\left[\tfrac{Z^0_T}{\gamma}\log\left(Z_T^0\right) + Z_T^0e_T +x\right]+\sup_{Z^0_T\in\cM^\lambda_a}\massE\left[-\tfrac{Z^0_T}{\gamma}\log\left(Z_T^0\right)-x\right]. 
      \end{aligned}    
    \end{equation*}
\end{proposition}

\begin{proof}
  By the special property of the exponential function we have that 
   \begin{equation}  \label{u(x+w)=eu(x)}
     u^{e_T}(x+w) = e^{-\gamma w}u^{e_T}(x).
   \end{equation}
  In particular, $u^{e_T}(x)=e^{-\gamma x}u^{e_T}(0)$, which follows that 
   $$ \lim_{x\to -\infty}u^{e_T}(x)=-\infty, \qquad  \lim_{x\to \infty}u^{e_T}(x)=0.  $$
  Since $u^{e_T}$ is concave, continuous and strictly increasing, there exists a solution of the equation $u^{e_T}(x-p)=u^0(x)$, denoted by $p(x)$. Again by \eqref{u(x+w)=eu(x)} we have that 
    $$ \exp\big(\gamma p(x)\big) u^{e_T}(x) = u^{e_T}\big(x-p(x)\big) = u^0(x). $$
  The assertion follows by a simple computation and Lemma \ref{u(x)=infE}. 
\end{proof}

\begin{corollary}
 Under the assumptions for Theorem \ref{maintheoremR}, the utility based pricing can be represented by the solution of dual problem on shadow markets, i.e., 
 \begin{align*}
  p(x)&= \inf_{Z^0_T\in\cM^\lambda_a}\massE\left[\tfrac{Z^0_T}{\gamma}\log\left(\tfrac{Z_T^0}{\gamma}\right) -\tfrac{Z_T^0}{\gamma}+ Z_T^0e_T\right]
               -\inf_{Z^0_T\in\cM^\lambda_a}\massE\left[\tfrac{Z^0_T}{\gamma}\log\left(\tfrac{Z_T^0}{\gamma}\right)-\tfrac{Z_T^0}{\gamma}\right]\\
      &=v\big(1;\widehat{S}(x; e_T)\big)-v\big(1;\widehat{S}(x)\big),
 \end{align*}
  where $\widehat{S}(x; e_T)$ is the generalized shadow price corresponding to the problem \eqref{PPP} with $x$ and $e_T$, while $\widehat{S}(x)$ is the one corresponding to the \eqref{PPP} with $x$ but without random endowment.
\end{corollary}

\begin{remark}
 The choice of the generalized shadow price will not alter the above result.
\end{remark}

\bibliography{RE_TC_SP_rev}

\bibliographystyle{plain}

\end{document}